\documentclass[11pt]{amsart}
\usepackage{amsmath}
\usepackage{amssymb}
\usepackage{harpoon}
\usepackage[dvips]{epsfig}
\theoremstyle{plain}
\newtheorem{theorem}{Theorem}[section]

\theoremstyle{definition}
\newtheorem{remark}[theorem]{Remark}

\numberwithin{equation}{section}

\begin{document}

\title[Inequalities Between Size and Charge for Bodies and Existence of Black Holes]
{Inequalities Between Size and Charge for Bodies and the Existence of Black Holes Due to Concentration of Charge}

\author[Khuri]{Marcus A. Khuri}
\address{Department of Mathematics\\
Stony Brook University\\
Stony Brook, NY 11794, USA}
\email{khuri@math.sunysb.edu}

\thanks{The author acknowledges the support of
NSF Grant DMS-1308753.}

\begin{abstract}
A universal inequality that bounds the charge of a body by its size is presented, and is proven as
a consequence of the Einstein equations in the context of initial data sets which satisfy an appropriate energy condition. We also present a general sufficient condition for the formation of black holes due to concentration of charge, and discuss the physical relevance of these results.
\end{abstract}
\maketitle

\section{Introduction}
\label{sec1} \setcounter{equation}{0}
\setcounter{section}{1}

It is well known that black holes of a fixed size can only support a certain amount of charge,
depending on the horizon area\footnote{This refers to apparent horizons. The same statement is not necessarily true for event horizons \cite{MurataReallTanahashi}.} \cite{DainJaramilloReiris,Gibbons,KhuriYamadaWeinstein}. Here we propose
a similar statement for arbitrary charged bodies which do not lie inside a black hole. Namely, let $\Omega$ be a compact spacelike hypersurface in a spacetime which satisfies a
suitable energy condition. If $\Omega$ lies in the domain of outer communication then there exists a universal constant $\mathcal{C}$ such that
\begin{equation}\label{0.1}
\mathrm{Charge}(\Omega)\leq \mathcal{C}\cdot \mathrm{Size}(\Omega),
\end{equation}
where a precise definition of size will be given later. Thus all bodies, from elementary particles to astronomical objects, can only support a fixed amount of charge depending on their size, or rather they must be sufficiently large depending on their charge. Similar results and inequalities have recently been obtained \cite{Dain,Khuri,Reiris} where the role of charge is replaced by angular momentum, that is
\begin{equation}\label{0.2}
\mathrm{AM}(\Omega)\leq\mathcal{C}\cdot \mathrm{Size}(\Omega)
\end{equation}
if $\Omega$ is not inside a black hole. The constant $\mathcal{C}$ in \eqref{0.1} will be a multiple of
$c^{2}/\sqrt{G}$, where $c$ is the speed of light and $G$ is the gravitational constant, and $\mathrm{Size}(\Omega)$ will be measured in units of length.

In \cite{Khuri}, it was shown that if the amount of angular momentum of a body sufficiently exceeds its size, then the body must be contained in a black hole. In this paper, we will also establish such a criterion for black hole existence focusing instead on the role of charge. More precisely, if the opposite inequality of \eqref{0.1} holds, then $\Omega$ must be contained in a black hole. It follows that concentration of charge alone can result in gravitational collapse. This statement is naturally motivated by intuition, since large amounts of charge are associated with strong electromagnetic fields, and
high concentration of matter fields is known to result in black hole formation. This last statement is referred to as the Hoop Conjecture \cite{Thorne}, and is related to the Trapped Surface Conjecture \cite{Seifert}. These conjectures have received much attention, although they are not fully resolved. Many works \cite{BeigOMurchadha,BizonMalecOMurchadha1,BizonMalecOMurchadha2,
Flanagan,Khuri0,Malec1,Malec2,Wald1} involve strong hypotheses including spherical symmetry or maximal slices, and some results \cite{Eardley,SchoenYau2,Yau} are not useful for initial data with low extrinsic curvature.

\section{Explicit Formulation}
\label{sec2} \setcounter{equation}{0}
\setcounter{section}{2}

Let $(M, g, k, E, B)$ be initial data for the Einstein-Maxwell equations which includes a Riemannian 3-manifold $M$ with complete metric $g$, symmetric 2-tensor $k$ denoting
extrinsic curvature, and vector fields $E$ and $B$ representing the electric and magnetic fields. It is assumed that the stress-energy tensor $T^{ab}=T^{ab}_{EM}+T^{ab}_{M}$ is decomposed into two parts, one for the electromagnetic field, $T_{EM}$, and the other for the remaining matter fields, $T_{M}$. If $n^a$ is the timelike unit normal to the slice, then $\mu=T^{ab}n_{a}n_{b}$ and $J^{i}=T^{ia}n_{a}$ represent the energy and momentum densities of all matter fields, whereas $\mu_{M}=T_{M}^{ab}n_{a}n_{b}$ and $J_{M}^{i}=T_{M}^{ia}n_{a}$ give the energy and momentum densities of the non-electromagnetic matter fields. The initial data must satisfy the constraints
\begin{align}\label{1}
\begin{split}
\mu_{M} &= \mu-\frac{1}{8\pi}\left(|E|^{2}+|B|^{2}\right),\\
J_{M} &= J +\frac{1}{4\pi}E\times B,
\end{split}
\end{align}
where
\begin{align}\label{1.0}
\begin{split}
\frac{16\pi G}{c^{4}}\mu &= R+(Tr_{g}k)^{2}-|k|^{2},\\
\frac{8\pi G}{c^{4}} J &= div(k-(Tr_{g}k)g),
\end{split}
\end{align}
with $R$ the scalar curvature of $g$, and $(E\times B)_{i}=\epsilon_{ijl}E^{j}B^{l}$ the cross product; here $\epsilon$ is the volume form of $g$. Recall also that the electric and magnetic fields are obtained from the field strength by $E_{i}=F_{in}$ and $B_{i}=-\frac{1}{2}\epsilon_{ijl}F^{jl}$. The following inequality will be referred to as the charged dominant energy condition
\begin{equation}\label{2}
\mu_{M}\geq|J_{M}|.
\end{equation}

Let $\Omega$ be a body, that is, a connected open subset of $M$ with compact closure and smooth boundary $\partial\Omega$. The sum of the squares of its electric and magnetic charges yields the square of total charge, which (using Gaussian units) is given by
\begin{equation}\label{2.1}
Q^{2}=\left(\frac{1}{4\pi}\int_{\Omega}div E d\omega_{g}\right)^{2}+\left(\frac{1}{4\pi}\int_{\Omega}div B d\omega_{g}\right)^{2}.
\end{equation}

We now describe measurements of the size of the body $\Omega$. In this regard, two definitions of radius will be important. Namely, in \cite{SchoenYau2} Schoen and Yau defined a radius $\mathcal{R}_{SY}(\Omega)$, that may be characterized as the largest (minor) radius among all tori that can be embedded in $\Omega$, and which was central to their condition for the existence of black holes due to compression of matter. In more detail, if $\Gamma$ is a simple closed curve that encloses a disk in $\Omega$, and $r$ is the greatest distance from $\Gamma$ with the property that all points within this distance combine to form an embedded torus in $\Omega$, then $\mathcal{R}_{SY}(\Omega)$ is equal to the maximum $r$ from any such curve $\Gamma$. Another important radius, the \'{O} Murchadha radius $\mathcal{R}_{OM}(\Omega)$, is defined \cite{OMurchadha} as the radius of the largest stable minimal surface that can be embedded in $\Omega$. Here, radius of the surface means the largest distance from a point in the surface to the boundary $\partial\Omega$, as measured by the induced metric on the surface. This is formulated most simply when $\partial\Omega$ is mean convex (has positive mean curvature), so that geometric measure theory guarantees the existence of many smooth least area surfaces contained in $\Omega$. In general, the \'{O} Murchadha radius gives a larger measure of size than the Schoen/Yau radius
\begin{equation}\label{2.2}
\mathcal{R}_{OM}(\Omega)\geq\mathcal{R}_{SY}(\Omega).
\end{equation}

Both radii measure well the size of a ball of radius $a$ in flat space. Namely, for this body $\mathcal{R}_{SY}=a/2$ and $\mathcal{R}_{OM}=a$. However, their measurement for a torus of major radius $a$ and minor radius $b$ is less accurate: $\mathcal{R}_{SY}=b/2$, $\mathcal{R}_{OM}=b$. In particular, for a torus, this measurement of size does not take into account the major radius. This leads to a problem if one tries to establish an inequality of the form \eqref{0.1}, with the notion of size given in terms of either of these radii. For instance, in the weak field limit a torus of large major radius $a$ but small minor radius $b$ could still support a large amount of charge, since its surface area and volume may be large, while the measure of its size in terms of the radii is small. For this reason, we choose a notion of size which incorporates surface area $|\partial\Omega|$ as well. That is, in the precise version of inequality \eqref{0.1}, size is defined by
\begin{equation}\label{2.3}
\mathrm{Size}(\Omega)=\frac{|\partial\Omega|}{\mathcal{R}(\Omega)},
\end{equation}
where $\mathcal{R}(\Omega)$ represents either the Schoen/Yau radius or the \'{O} Murchadha radius. Lastly, it should be mentioned that the radius $\mathcal{R}_{OM}$ gives an accurate measurement for highly dense spherical bodies \cite{Dain, OMurchadha}. Thus even in a strong gravitational field, this measurement is on the order of the area radius. Nevertheless it would be desirable to eventually replace the role of the Schoen/Yau and \'{O} Murchadha radii with a more simple quantity. However our results here depend on theorems proved for these more complicated measurements, and it is not clear how to proceed without them.

With these notions of charge and size, we may formulate a rigorous description of inequality \eqref{0.1}, except for the constant $\mathcal{C}$ which will be given in the next section. In order to
articulate the black hole existence result, we must replace event horizons
with the quasilocal notion of apparent horizons. This is due to the fact that event horizons
cannot be located in initial data without knowledge of the full spacetime development, whereas apparent horizons may be identified directly from the initial data. Let $S\subset M$ be a 2-surface and let
$\theta_{\pm}:=H_{S}\pm Tr_{S}k$
be the null expansions, which assess potency of the gravitational field;
here $H_{S}$ denotes mean curvature with regard to the unit
outward normal. Geometrically, the null expansions arise from the first variation of area in the outward future ($\theta_{+}$) and outward past ($\theta_{-}$) null directions, and as such they are the rate of change of area for a shell of radiation
discharged by $S$ in these directions. If $\theta_{+}< 0$ ($\theta_{-}< 0$), then $S$ is called a future (past) trapped surface, and the gravitational field is considered to be strong in this vicinity.
The boundaries of future (past) trapped regions are referred to as future (past) apparent horizons, and
solve the equation $\theta_{+}=0$
($\theta_{-}=0$). Cosmic censorship implies that apparent horizons are typically contained inside black holes \cite{Wald}, and thus in many situations they may be used in place of event horizons. We will show that if the opposite inequality of \eqref{0.1} is valid, then an apparent horizon must exist within the initial data.

\section{Inequalities Relating Size and Charge of Bodies}
\label{sec3} \setcounter{equation}{0}
\setcounter{section}{3}

In this section, inequalities of the form \eqref{0.1} will be established, both in the maximal case ($Tr_{g}k=0$) and in the general case. The inequality obtained in the maximal case is stronger, as the universal coefficient $\mathcal{C}$ in this case is smaller. However the inequality obtained for general initial data will be used to obtain the criterion for black hole existence, described in the next section. We begin with an important observation which
will be used in both cases. Let $\Omega$ be a body as described in the previous section, then the total charge may be estimated in terms of the energy and momentum densities as follows
\begin{align}\label{3}
\begin{split}
Q^{2}
=&\left(\frac{1}{4\pi}\int_{\partial\Omega}E\cdot\nu d\sigma_{g}\right)^{2}
+\left(\frac{1}{4\pi}\int_{\partial\Omega}B\cdot\nu d\sigma_{g}\right)^{2}\\
\leq&\frac{|\partial\Omega|}{16\pi^{2}}
\int_{\partial\Omega}\left(|E|^{2}+|B|^{2}\right)d\sigma_{g}\\
=&\frac{|\partial\Omega|}{16\pi^{2}}\int_{\partial\Omega}\left[\left(|E|^{2}+|B|^{2}
-8\pi\mu+8\pi|J_{M}|\right)+8\pi(\mu-|J_{M}|)\right]d\sigma_{g}\\
\leq&\frac{|\partial\Omega|}{2\pi}\int_{\partial\Omega}(\mu-|J_{M}|)d\sigma_{g},
\end{split}
\end{align}
where in the last step the charged dominant energy condition \eqref{2} was used, and
$\nu$ is the unit outer normal to $\partial\Omega$.

\begin{theorem}\label{thm1}
Let $(M,g,k,E,B)$ be a maximal $(Tr_{g}k=0)$ initial data set for the Einstein-Maxwell equations. Then for any body $\Omega\subset M$ with constant energy density $\mu$ and satisfying
the charged dominant energy condition \eqref{2}, the following inequality holds
\begin{equation}\label{4}
|Q|\leq\sqrt{\frac{c^{4}}{12G}}\frac{|\partial\Omega|}{\mathcal{R}(\Omega)},
\end{equation}
where $\mathcal{R}(\Omega)$ denotes the Schoen/Yau radius $\mathcal{R}_{SY}(\Omega)$. Moreover, if in addition the boundary $\partial\Omega$ is mean convex, then $\mathcal{R}(\Omega)$ denotes the \'{O} Murchadha radius $\mathcal{R}_{OM}(\Omega)$.
\end{theorem}

\begin{proof}
It suffices to estimate the integral on the right-hand side of \eqref{3}. We have
\begin{equation}\label{4.1}
\int_{\partial\Omega}(\mu-|J_{M}|)d\sigma_{g}
\leq\int_{\partial\Omega}\mu d\sigma_{g}=\mu|\partial\Omega|.
\end{equation}
In light of the maximal assumption and the constancy of $\mu$, Theorem 1 of \cite{SchoenYau2} may be applied to yield
\begin{equation}\label{5}
\mu\leq\frac{\pi c^{4}}{6G\mathcal{R}_{SY}(\Omega)^{2}}.
\end{equation}
It follows from \eqref{3} that
\begin{equation}\label{6}
Q^{2}\leq\frac{c^{4}}{12G}\frac{|\partial\Omega|^{2}}{\mathcal{R}_{SY}(\Omega)^{2}}.
\end{equation}

Now consider the case when the boundary $\partial\Omega$ is mean convex. It was pointed out
in \cite{GallowayOMurchadha}, that under this additional hypothesis, the estimate \eqref{5}
holds with the \'{O} Murchadha radius
\begin{equation}\label{5.1}
\mu\leq\frac{\pi c^{4}}{6G\mathcal{R}_{OM}(\Omega)^{2}}.
\end{equation}
It follows that \eqref{4} holds with the \'{O} Murchadha radius. Note that \eqref{2.2} implies that this is a better
result than the estimate with the Schoen/Yau radius.
\end{proof}

We will now obtain an inequality between the size and charge of bodies without the maximal assumption. Here we will employ a technique developed by Schoen and Yau in \cite{SchoenYau1, SchoenYau2}, which reduces certain problems for general initial data back to the maximal setting. The idea is that in the maximal setting, nonnegative scalar curvature $R\geq 0$ is guaranteed from the dominant energy condition, and it is this nonnegativity which is fundamental for establishing many geometric inequalities, such as the positive mass theorem or \eqref{5} in the proof of Theorem \ref{thm1}. Thus, it is natural to deform the initial data metric $g$ to a new unphysical metric $\overline{g}$ whose scalar curvature satisfies $\overline{R}\geq 0$, at least in a weak sense. This is accomplished in \cite{SchoenYau1} by setting $\overline{g}_{ij}=g_{ij}+\nabla_{i}f\nabla_{j}f$, which is the induced metric on the graph $t=f(x)$ in the 4-dimensional product manifold $\mathbb{R}\times M$, where $f$ satisfies the so called Jang equation
\begin{equation}\label{11}
\left(g^{ij}-\frac{f^{i}f^{j}}{1+|\nabla f|^{2}}\right)\left(\frac{\nabla_{ij}f}
{\sqrt{1+|\nabla f|^{2}}}-k_{ij}\right)=0,
\end{equation}
with $f^{i}=g^{ij}\nabla_{j}f$. The purpose of this equation is to guarantee that the scalar curvature of $\overline{g}$ is weakly nonnegative, in fact it is given by the following formula \cite{BrayKhuri1,BrayKhuri2,SchoenYau1}
\begin{equation}\label{12}
\overline{R}=\frac{16\pi G}{c^{4}}(\mu-J(v))+
|h-k|_{\overline{g}}^{2}+2|q|_{\overline{g}}^{2}
-2div_{\overline{g}}(q).
\end{equation}
Here $div_{\overline{g}}$
is the divergence operator, $h$ is the second fundamental form of the graph, and
\begin{equation}\label{13}
v_{i}=\frac{f_{i}}{\sqrt{1+|\nabla f|^{2}}},\text{
}\text{ }\text{ }\text{ }
q_{i}=\frac{f^{j}}{\sqrt{1+|\nabla f|^{2}}}(h_{ij}-k_{ij}).
\end{equation}
If the dominant energy condition $\mu\geq|J|$ is valid, then each term on the right-hand side of \eqref{12} is clearly nonnegative, except perhaps the divergence term; hence we
may view $\overline{R}$ as being weakly nonnegative, which is sufficient for most purposes.

Let us assume that the Jang equation \eqref{11} possesses a smooth solution in $\Omega$. Measurement of the accumulation of matter fields, which is needed to estimate the right-hand side of \eqref{3}, may be obtained by assessing the concentration of scalar curvature for the unphysical metric $\overline{g}$. This in turn may be accomplished by estimating
the principal Dirichlet eigenvalue
\begin{equation}\label{14}
\lambda_{1}=\frac{\int_{\Omega}\left(|\overline{\nabla}\phi|^{2}
+\frac{1}{2}\overline{R}\phi^{2}\right)d\omega_{\overline{g}}}
{\int_{\Omega}\phi^{2}d\omega_{\overline{g}}}
\end{equation}
of the operator $\Delta_{\overline{g}}-\frac{1}{2}\overline{R}$,
where $\phi$ is the principal eigenfunction.
From the weak nonnegativity of the scalar curvature, we may integrate by parts and use the two nonnegative terms $|\overline{\nabla}\phi|^{2}$ and $|q|_{\overline{g}}^{2}\phi^{2}$ to find
\begin{equation}\label{16}
\lambda_{1}\geq\frac{8\pi G}{c^{4}}\min_{\Omega}(\mu-|J|)=:\Lambda.
\end{equation}
Here we have also used the fact that $|v|\leq 1$, so that $\mu-J(v)\geq\mu-|J|$.

The estimate \eqref{16} for the principal eigenvalue allows an application of Proposition 1 of \cite{SchoenYau2} (if $\Lambda\neq 0$), from which we obtain
\begin{equation}\label{17}
\overline{\mathcal{R}}_{SY}(\Omega)\leq \sqrt{\frac{3}{2}}\frac{\pi}{\sqrt{\Lambda}},
\end{equation}
where the radius $\overline{\mathcal{R}}_{SY}$ pertains to $\overline{g}$. Since the metric $\overline{g}$ measures lengths larger than does the metric $g$, it follows that $\overline{\mathcal{R}}_{SY}
\geq\mathcal{R}_{SY}$. Furthermore, let $\psi\in C^{\infty}(\Omega)$ be an arbitrary positive function, then dividing and multiplying $\Lambda^{-1}$ by
$\int_{\Omega}\psi d\omega_{g}\left(\int_{\Omega}(\mu-|J|)\psi d\omega_{g}
\right)^{-1}$ produces
\begin{equation}\label{18}
\Lambda^{-1}\leq\frac{c^{4}\mathcal{C}_{0}}{8\pi G}\frac{\int_{\Omega}\psi d\omega_{g}}
{\int_{\Omega}(\mu-|J|)\psi d\omega_{g}},
\end{equation}
where
$\mathcal{C}_{0}=
\frac{\max_{\Omega}\left(\mu-|J|\right)}
{\min_{\Omega}\left(\mu-|J|\right)}$
when $\mu-|J|>0$ in $\Omega$, and $\mathcal{C}_{0}=\infty$ when $\mu-|J|=0$ somewhere in $\Omega$. Therefore
\begin{equation}\label{20}
\int_{\Omega}(\mu-|J|)\psi d\omega_{g}\leq
\frac{3\pi c^{4}\mathcal{C}_{0}}{16G}\frac{\int_{\Omega}\psi d\omega_{g}}
{\mathcal{R}_{SY}(\Omega)^{2}},
\end{equation}
and this leads to a general correlation between the size and charge of bodies. Before stating this result, we record integral forms of the charged dominant energy condition and an `enhanced' energy condition on the boundary
\begin{equation}\label{20.1}
\int_{\partial\Omega}\left(\mu_{M}-|J_{M}|\right)d\sigma_{g}\geq 0,\text{ }\text{ }\text{ }\text{ }\text{ }\text{ }\text{ }\text{ }\text{ }\text{ }\int_{\partial\Omega}\left(\mu_{M}-|J|\right)d\sigma_{g}\geq 0,
\end{equation}
and define a constant $\mathcal{C}_{1}^{2}=\mathcal{C}_{0}\frac{\max_{\partial\Omega}(\mu-|\tilde{J}|)}
{\min_{\Omega}(\mu-|J|)}$ where $\tilde{J}$ is $J_{M}$ or $J$ depending on whether the first or second condition of \eqref{20.1} is satisfied, respectively.

\begin{theorem}\label{thm2}
Let $(M,g,k,E,B)$ be an initial data set for the Einstein-Maxwell equations,
which contains no compact apparent horizons. Assume that either $M$ is asymptotically flat, or has a strongly untrapped boundary, that is $H_{\partial M}>|Tr_{\partial M}k|$. Then for any body $\Omega\subset M$ satisfying the dominant energy condition $\mu\geq|J|$, and one of the two energy conditions \eqref{20.1} on $\partial\Omega$, the following inequality holds
\begin{equation}\label{21}
|Q|\leq\mathcal{C}_{1}\sqrt{\frac{3c^{4}}{32G}}\frac{|\partial\Omega|}{\mathcal{R}_{SY}(\Omega)}.
\end{equation}
\end{theorem}

\begin{proof}The assumptions concerning the asymptotics of $M$ or its boundary imply that a strongly untrapped 2-surface is present in the initial data, and hence the Dirichlet problem \cite{SchoenYau2} for the Jang equation admits a solution with $f=0$ on such a surface. In addition, $f$ must be a smooth solution due to the lack of apparent horizons. The arguments above now apply, so that \eqref{20} is valid.

Let us now assume that the first energy condition of \eqref{20.1} holds. Then
\eqref{3} is valid and we find that
\begin{equation}\label{22}
Q^{2}
\leq\frac{|\partial\Omega|}{2\pi}\int_{\partial\Omega}(\mu-|J_{M}|)d\sigma_{g}.
\end{equation}
Furthermore by choosing
\begin{equation}\label{23}
\psi=\frac{\int_{\partial\Omega}(\mu-|J_{M}|)d\sigma_{g}}
{\int_{\Omega}(\mu-|J|)d\omega_{g}}\leq\frac{\max_{\partial\Omega}(\mu-|J_{M}|)}
{\min_{\Omega}(\mu-|J|)}\frac{|\partial\Omega|}
{|\Omega|},
\end{equation}
we obtain
\begin{align}\label{24}
\begin{split}
\int_{\partial\Omega}(\mu-|J_{M}|)d\sigma_{g}\leq &
\frac{3\pi c^{4}\mathcal{C}_{1}^{2}}{16G}\frac{|\partial\Omega|}{\mathcal{R}_{SY}(\Omega)^{2}}
\end{split}
\end{align}
from \eqref{20}. Together, \eqref{22} and \eqref{24} yield the stated conclusion.

Alternatively, if the second energy condition of \eqref{20.1} holds, then add and subtract $|J|$ instead of $|J_{M}|$ in \eqref{3} to obtain
\begin{equation}\label{22.1}
Q^{2}
\leq\frac{|\partial\Omega|}{2\pi}\int_{\partial\Omega}(\mu-|J|)d\sigma_{g}.
\end{equation}
By choosing
\begin{equation}\label{23.1}
\psi=\frac{\int_{\partial\Omega}(\mu-|J|)d\sigma_{g}}
{\int_{\Omega}(\mu-|J|)d\omega_{g}}\leq\frac{\max_{\partial\Omega}(\mu-|J|)}
{\min_{\Omega}(\mu-|J|)}\frac{|\partial\Omega|}
{|\Omega|},
\end{equation}
a similar argument yields \eqref{21}.
\end{proof}

\begin{remark}
The first energy condition in \eqref{20.1} is satisfactory, as it is a weaker version of the well known energy condition \eqref{2} (used for instance in the positive mass theorem with charge) which states that the non-electromagnetic matter fields satisfy the dominant energy condition. On the other hand, the second energy condition in \eqref{20.1} is not well motivated, and we view it as a technical assumption that could perhaps be removed with further investigation. It should also be pointed out that the constant in \eqref{21} is independent of the particular matter model, as long as it satisfies the appropriate energy condition.
\end{remark}


This result generalizes Theorem \ref{thm1}, which involves strong hypotheses such as the assumption of maximal data and constant matter density. Note, however, that the inequality of Theorem \ref{thm2} is weaker than that of Theorem \ref{thm1}, since the constant $\mathcal{C}_{1}\sqrt{\frac{3c^{4}}{32G}}$ appearing in \eqref{21} is generally larger than the constant $\sqrt{\frac{c^{4}}{12G}}$ of \eqref{4}. Thus, although two undesirable hypotheses have been removed, the resulting inequality is not optimal;
the problem of finding the optimal constant is currently being investigated by Dain et al. \cite{AngladaDainOrtiz} in the context of spherical symmetry. We also mention that the charged bodies constructed by Bonnor in \cite{Bonnor}, all of which satisfy \eqref{21}, could potentially be useful in this pursuit.
It turns out that the difference in the constants just described is related to a black hole existence result which we now explain.

\section{Criteria for Black Hole Formation}
\label{sec4} \setcounter{equation}{0}
\setcounter{section}{4}

The dependence of Theorem \ref{thm2} on solutions of the Jang equation \eqref{11}
inherently produces a black hole existence result. This is due to the fact
that solutions are regular except possibly at apparent horizons, where the
graph $t=f(x)$ tends to blow-up in the form of a cylinder (see \cite{HanKhuri}, \cite{SchoenYau1}). In
other words, if it can be shown that the Jang equation does not possess a regular solution, then an apparent horizon must be present in the initial data.
This method for producing black holes was initially
used by Schoen and Yau in \cite{SchoenYau2}. Here we will use it to obtain a criterion
for black hole existence due to concentration of charge.

\begin{theorem}\label{thm3}
Let $(M,g,k,E,B)$ be an initial data set for the Einstein-Maxwell equations,
such that either $M$ is asymptotically flat, or has a strongly untrapped boundary, that is $H_{\partial M}>|Tr_{\partial M}k|$. If $\Omega\subset M$ is a body satisfying the dominant energy condition $\mu\geq|J|$ and one of the two energy conditions \eqref{20.1}, with
\begin{equation}\label{25}
|Q|>\mathcal{C}_{1}\sqrt{\frac{3c^{4}}{32G}}\frac{|\partial\Omega|}{\mathcal{R}_{SY}(\Omega)},
\end{equation}
then $M$ contains an apparent horizon of spherical topology which encloses a region that intersects $\Omega$.
\end{theorem}

\begin{proof}
As in the proof of Theorem \ref{thm2}, the assumptions on the boundary of $M$ or its asymptotics imply the existence of a solution to the Dirichlet boundary
value problem for the Jang equation, with $f=0$ on $\partial M$ or on an appropriate coordinate sphere in the asymptotic end. If the solution were regular, then by Theorem \ref{thm2} the opposite inequality of \eqref{25} would hold. Since this is not the case,
we conclude that the solution is not regular, and hence the existence of an
apparent horizon is guaranteed. Moreover, among apparent horizons arising from the blow-up of Jang's
equation, there is at least one (which is outermost) with spherical topology \cite{SchoenYau1}. If the region which is enclosed by this apparent horizon does not intersect $\Omega$, then the solution of Jang's equation is smooth over $\Omega$. This would then imply by the proof of Theorem \ref{thm2} that \eqref{21} holds, yielding a contradiction.
\end{proof}

Whether concentrated charge leading to gravitational collapse is a naturally occurring phenomenon, seems to be an intriguing open question. In the next section we will comment on some physical aspects of this problem.
As for the theoretical part, it should be pointed out that the proposed criteria is not
satisfied in the maximal case, since as mentioned previously, the universal constant in inequality \eqref{4} is smaller than
the constant in \eqref{21} and \eqref{25}. A similar relation between the maximal and nonmaximal cases holds with regards to
Schoen and Yau's condition for black hole creation. Therefore, sufficient amounts of extrinsic curvature are needed for the Schoen/Yau condition, as well as the hypotheses of Theorem \ref{thm3}, to be fulfilled.

With this intuition, we now show how to construct examples of initial data satisfying the conditions for black hole existence. Fix an asymptotically flat metric $g$ on $M\simeq\mathbb{R}^{3}$, and set $\Omega=B_{1}(0)$ to be the unit ball. Let $(e_{1},e_{2},e_{3})$ be an orthonormal frame and set $k_{ij}:=k(e_{i},e_{j})=0$ for all $(i,j)\neq(1,1), (2,2)$, and $k_{11}=k_{22}=\beta^2$ for some parameter $\beta>0$. It follows that
\begin{equation}\label{25.1}
\mu =\frac{c^{4}}{16\pi G}(R+2k_{11}k_{22})\sim\beta^{4},\text{ }\text{ }\text{ }\text{ }\text{ }\text{ }|J|\sim\beta^2.
\end{equation}
For large $\beta$ this yields $\mathcal{C}_{1}\sim 1$, and also $\mu_{M}>|J|$ if we choose $E,B\sim\beta$. It may be arranged so that $div E=div B=4\pi\beta$ in $\Omega$, which implies that $|Q|=\alpha\sqrt{2}|\Omega|$. Since the right-hand side of \eqref{25} is independent of $\beta$, it follows that \eqref{25} is satisfied for large $\beta$. Lastly, by extending $k$, $E$, and $B$ outside of $\Omega$ to be asymptotically flat, all the hypotheses of Theorem \ref{thm3} are satisfied and an apparent horizon must be present. Since the constant $2\beta^{2}$ represents $Tr_{g}k$, we conclude that large traces of extrinsic curvature facilitate the formation of apparent horizons in this setting. Note that there is quite a bit of freedom in this construction, since for instance the metric $g$ is essentially arbitrary. Of course, this is an abstract and ad hoc procedure which is not of physical interest, but it shows that there are many configurations satisfying the criteria of this result. It should also be pointed out that besides the Maxwell field, the matter models present in this construction are not given explicitly. However, it seems that a combination of dust (or perfect fluid) with a charged scalar field might fit. Finally we mention that since $\mu-|J|\sim\beta^4$, and the radius $\mathcal{R}_{SY}(\Omega)$ is fixed, this data also satisfies the hypotheses of the Schoen/Yau black hole existence criterion \cite{SchoenYau2}. Moreover a similar construction in axisymmetry, with the role of charge replaced by angular momentum, yields examples of data satisfying the criteria of \cite{Khuri}.


\section{Physical Relevance}
\label{sec5} \setcounter{equation}{0}
\setcounter{section}{5}

The inequalities between size and charge for bodies, as well as the black hole existence criterion proven above, are predictions of Einstein's theory and hence should be contrasted with observational evidence and other theories. Let us consider bodies which are approximately spherical in shape, so that the ratio of boundary area to radius is on the order of the radius
$\mathcal{R}$. Then in general terms, what we have shown is that for stable bodies
\begin{equation}\label{26}
|Q|\lesssim \frac{c^{2}}{\sqrt{k_{e}G}}\mathcal{R},
\end{equation}
and that if the opposite inequality holds then the body should undergo gravitational collapse. Here $\lesssim$ should be interpreted in terms of order of magnitude, and $k_{e}\approx9\times10^{9} Nm^{2}C^{-2}$ is Coulomb's constant so that \eqref{26} is expressed in SI units, as opposed to Gaussian units used in previous sections.

Consider now an electron. It has a classical radius of $\mathcal{R}_e\approx 2.8\times 10^{-15}m$. Moreover, since $G\approx 6.67\times 10^{-11}Nm^{2}kg^{-2}$ and $c\approx 3\times 10^{8}ms^{-1}$ it follows that
\begin{equation}\label{27}
\frac{c^{2}}{\sqrt{k_{e}G}}\mathcal{R}_e
\approx 100 C.
\end{equation}
Therefore, since the charge of an electron $|Q_e|\approx 1.6\times 10^{-19}C$, we find that \eqref{26} is satisfied.

According to the principle of charge quantization, the charge of a body is an integer multiple of the elementary charge (charge of an electron). Thus, $|Q_e|$ is the smallest amount of charge that a body can possess. Using this fact in \eqref{26}, we find that the classical theory imposes the following minimum size for a body
\begin{equation}\label{27.1}
\mathcal{R}_{0}= \frac{\sqrt{k_{e}G}}{c^{2}}|Q_e|\approx 1.4\times 10^{-36}m,
\end{equation}
which is on the order of the Planck length $l_{p}=\left(\frac{G\hbar}{c^{3}}\right)^{1/2}\approx 1.6\times 10^{-35}m$. It then appears to be a remarkable self consistency of
the Einstein field equations that they predict a minimum
length on the order of magnitude of the Planck length, if
we assume the principle of charge quantization.

On the other hand, we may consider bodies of astronomical scale such as stars. The study of the effects of electric charge in isolated gravitating systems goes back to Rosseland \cite{Rosseland} and Eddington \cite{Eddington}. It was shown that since electrons are rather less massive than protons, electrons tend to escape more frequently, as part of the solar wind. This induces a net
positive charge in the star, which then yields an attractive force on electrons trying to escape.
Eventually an equilibrium of these forces is established, resulting in a net positive charge
on the order of $\sim100 (M/M_{\odot})C$ \cite{HarrisonBally}, where $M$ is the mass of the star.
Thus, for typical stars, net charge is sufficiently small to be considered insignificant, and they certainly satisfy inequality \eqref{26}. However, as pointed out by Witten \cite{Witten}, it is theoretically possible to have stars made of absolutely stable strange quark matter. These are highly dense bodies, which have masses and radii similar to those
of neutron stars. They are also capable of possessing large amounts of charge \cite{NWMU}, and thus are candidates to violate \eqref{26}. Consider such a star with charge $|Q|=10^{20} C$ and radius $\mathcal{R}=10^{4}m$ as considered in \cite{NWMU}.  We have
\begin{equation}\label{28}
\frac{c^{2}}{\sqrt{k_{e}G}}\mathcal{R}
\approx 10^{21} C,
\end{equation}
so that \eqref{26} is still satisfied, although it is nearly violated. Moreover, the black hole existence criterion associated with \eqref{26} asserts that a star of this radius can
only support a charge of $|Q|\sim 10^{20}C$, beyond which the system will collapse to
form a black hole; this is consistent with the findings of \cite{REMLZ}, obtained numerically
with different methods. Lastly, we mention that magnetic charge is also included on the left hand side of \eqref{26}, and thus it would be interesting to contrast the above results with empirical evidence associated with magnetic charge.

\end{document}